\documentclass[10pt,journal]{IEEEtran}
\usepackage{amsmath,amssymb}
\usepackage{epsfig}
\usepackage{graphics}
\usepackage{lscape}
\usepackage[all]{xy}
\usepackage{enumerate}
\usepackage{psfrag}
\usepackage{rotating}
\usepackage{wrapfig}
\usepackage{multicol}
\usepackage[boxed]{algorithm2e}

\newcommand{\R}{\mathbb{R}}
\newcommand{\Z}{\mathbb{Z}}

\newcommand{\Su}{\mathcal{S}}
\newcommand{\Vo}{\mathcal{V}}
\newcommand{\Sc}{\mathcal{S}\mathcal{C}}
\newcommand{\PP}{\mathcal{P}}
\newcommand{\cc}{\pmb{c}}
\newcommand{\bb}{\pmb{b}}
\newcommand{\xx}{\pmb{x}}
\newcommand{\uu}{\pmb{u}}
\newcommand{\vv}{\pmb{v}}
\newcommand{\xo}{\pmb{x_0}}
\newcommand{\NN}{\mathcal{N}}

\newtheorem{prop}{Proposition}

\ifCLASSINFOpdf
\else
\fi
\begin{document}
\title{Constructive spherical codes on layers of flat tori}

\author{Cristiano Torezzan, Sueli I. R. Costa, Vinay A. Vaishampayan%
\thanks{C. Torezzan is with the School of Applied Sciences, University of Campinas, Brazil (e-mail: cristiano.torezzan@fca.unicamp.br. S.I.R. Costa is with the Institute of Mathematics, University of Campinas, Brazil (e-mail: sueli@ime.unicamp.br. V. A. Vaimshampayan is with AT\&T Shannon Laboratory, 
Florham Park, NJ, USA (e-mail:vinay@research.att.com). 

This work was supported in part by S\~ao Paulo Research Foundation  FAPESP and National Counsel of Technological and Scientific Development - CNPq - Brazil. 

The material in this paper was presented in part at the IEEE International Symposium on Information Theory, Seoul, Korea, 2009.}

}


%\markboth{IEEE Transactions on Information Theory,~Vol.~?, No.~?, ?????~2011}%
%{Spherical codes on torus layers}

\maketitle

\begin{abstract}
%\boldmath
A new class of spherical codes is constructed by selecting a finite subset of flat tori from a foliation of the unit sphere ${\cal S}^{2L-1}\subset \R^{2L}$ and designing a structured codebook on each torus layer. The resulting spherical code can be the image of a lattice restricted to a specific hyperbox in $\R^L$ in each layer. Group structure and homogeneity, useful for efficient storage and decoding, are inherited from the underlying lattice codebook. A systematic method for constructing such codes are presented and, as an example, the Leech lattice is used to construct a spherical code in $R^{48}$. Upper and lower bounds on the performance, the asymptotic packing density and a method for decoding are derived.
\end{abstract}


% \fi
%\IEEEpeerreviewmaketitle

\section{Introduction}
The problem of  placing points on the unit Euclidean sphere of a given dimension has attracted the attention of engineers, mathematicians and scientists and has relevance to many diverse fields of science and engineering. In communication theory, point sets on the unit sphere are useful for communica\-ting over a Gaussian channel and are a natural generalization of phase shift keyed signal sets (PSK) to dimensions greater than two. The point placement problem in this case is formulated as a packing problem in which the objective is to pack as many spherical caps of given radius as possible on the sphere. The dual to the packing problem is the covering problem, useful in facility location, in which the problem is to minimize the number of spherical caps of a given radius so that every point on the sphere is covered. In quantization, point sets on the sphere form a key component in shape-gain vector quantizers \cite{GershoGray}. Point sets on the sphere with special properties are known as spherical $t$-designs. An overview on spherical codes and its properties can be found in  \cite{zino, ConwaySloane} and  two families of asymptotically dense codes are presented in  \cite{Hamkins1, Hamkins2}. Lists of good spherical packings, coverings and designs can be found online at \cite{Sloane:Spherical}.

In this paper, we describe a new method for constructing spherical codes for the communication problem. While it is important to maximize the packing density, additional practical considerations such as storage and easy decoding are also important. The codes presented here have low construction and decoding complexity and, for not asymptotically small distances, have comparable performance to some well known apple peeling  \cite{Gamal}, wrapped  \cite{Hamkins1} and laminated  \cite{Hamkins2} codes.

The paper is organized as follows. The flat tori foliation of the sphere is introduced in Sec. \ref{sec:FT}. Sec. \ref{sec:TLSC} describes our proposal of \textit{Torus Layes Spherical Codes (TLSC)} which is based on a foliation of the unit sphere in $\R^{2L}$. In Sec.  \ref{TLSC4} an example of TLSC in $R^4$, which is cyclic in each layer is presented with comparisons to some well known constructions. In Sec.  \ref{group} it is described how lattices with good packing density in $R^L$ can be used to construct TLSC in $R^{2L}$. As an example the Leech Lattice is used to construct a spherical code in $R^{48}$. Upper and lower bounds on the performance and asymptotic packing density of the TLSC are derived Sec. \ref{sec:Bounds}. Finally a decoding method for our codes is described in Sec.  \ref{sec:Decoding}.

\section{Foliation of the Sphere by Flat Tori}
\label{sec:FT}

The unit sphere $S^{2L-1} \subset \R^{2L} $ can be foliated with flat tori (Clifford Tori) \cite{BergerGostiaux}, \cite{CMAP:2004} as follows. For each unit {\textit{radius} vector $\cc = (c_{1},c_{2},..,c_{L}) \in \R^{L}, c_i > 0$, and $\pmb{u}=(u_1,u_2,\ldots,u_L) \in \R^L$, let $\Phi_{\cc}:\R^L \rightarrow \R^{2L}$ be defined by
\begin{equation}
\small{
\Phi_{\cc} (\pmb u)=\left(c_{1}\left(\cos \frac{u_{1}%
}{c_{1}},\sin \frac{u_{1}}{c_{1}}\right),\dots,c_{L}\left(\cos \frac{u_{L}}{c_{L}},\sin \frac{u_{L}}{c_{L}}\right)\right).
}
\label{phi}
\end{equation}
The image of $\Phi _{\cc}$ is the torus $ T_{\cc}$, a flat $L$-dimensional surface on the unit sphere $S^{2L-1}$.  $\Phi_{\cc}$ is an embedding of the flat torus $T_{\pmb{c}}
$, generated by the hyperbox:
\begin{equation}
\label{para}
\PP_{\cc} = \{\uu \in \R^{L}; 0 \leq u_{i} < 2 \pi c_{i}\}, \ \ 1\leq i\leq L.
\end{equation}

Note also that each vector of $S^{2L-1}$ belongs to one, and only one, of these flat tori, some of which may be degenerate\footnote{A degenerate torus $T_{\cc}$ is one for which some $c_i=0$, $i=1,2,\ldots,L$.}. 

\begin{figure}[h!]
	\centering
		\includegraphics[scale=0.35]{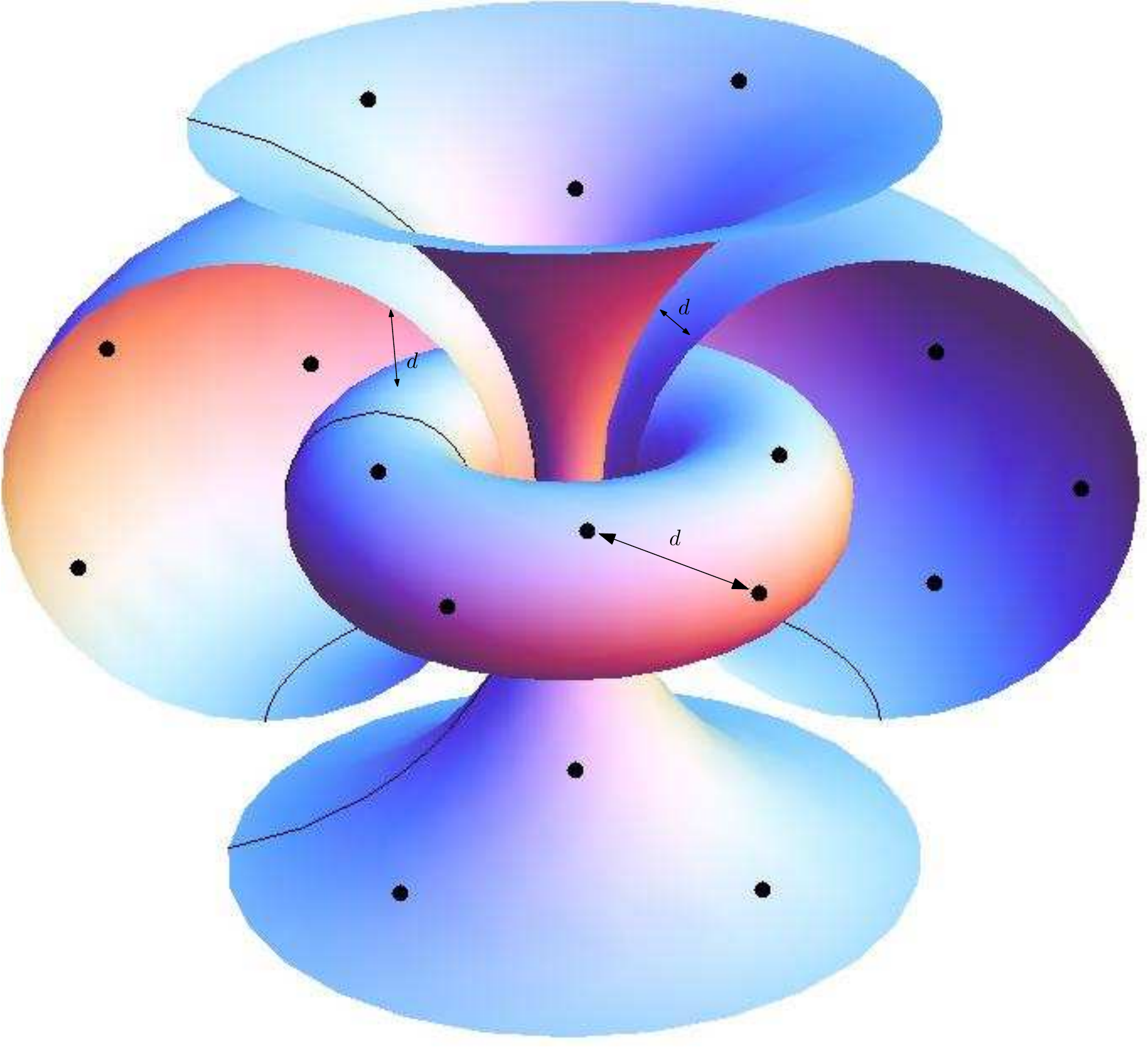}
		\caption{Ilustration of a torus layer spherical code in dimension four, projected in $\R^3$. The codewords belong to the surface of flat tori.}
  	 \label{fig:decombp}
\end{figure}

%$\xx = \{x_1, x_2, \cdots, x_L \}$ of $S^{2L-1}$ can be written as 
%\begin{eqnarray}
%\xx & = &  \left( \gamma_1 \left( \frac{x_1}{\gamma_1}, \frac{x_2}{\gamma_1} \right),  \hdots , \gamma_k \left( \frac{x_{2L-1}}{\gamma_L}, \frac{x_{2L}}{\gamma_L} \right)
% \right) = \nonumber \\
% & = &  \left( \gamma_1\left(\cos{\frac{\theta_{1}}{\gamma_1}}, \sin{\frac{\theta_{1}}{\gamma_1}} \right), \hdots ,
%\gamma_L \left(\cos{\frac{\theta_{L}}{\gamma_L}},\sin{ \frac{\theta_{L}}{\gamma_L}} \right) \right). \nonumber 
% \end{eqnarray}
%where,
% \begin{eqnarray}
% \gamma_i   & = & \sqrt{x_{2i-1}^2+x_{2i}^2}, \ \ 1 \leq i \leq L \nonumber \\
% \theta_{i} & = & \arccos{ \left( \dfrac{x_{2i-1}}{\gamma_i} \right) } \gamma_i,  \ \ 1 \leq i \leq L \nonumber.
% \end{eqnarray}
%this means that $\xx$ 

The Gaussian curvature of a torus $T_{\cc}$ is zero   \cite{BergerGostiaux} and $T_{\cc}$ can be cut and flattened into an $L$-dimensional box, $\PP_{\cc}$, just as a cylinder in $\R^3$ can be cut and flattened into a $2$-dimensional rectangle. Since the inner product $\langle \partial \Phi_{\cc}/\partial u_i, \partial \Phi_{\cc}/\partial u_j \rangle=\delta_{ij}$, the application $\Phi_{\cc}$ is a local isometry, which means that any measure of length, area and $L$-dimensional volume on $T_{\cc}$ is the same of the corresponding pre-image in the $L$-dimensional hyperbox $\PP_{\cc}$ (\ref{para}).

We say that the family of flat tori $T_{\cc}$ and their degenerations, with $\cc = (c_{1},c_{2},..,c_{L})$, $ \left\|  \cc  \right\|  =1$, $c_{i}  \geq 0$, defined above is a foliation on the unit sphere of $S^{2L-1}\subset \R^{2L}.$

Let $T_{\bb}$ and $T_{\cc}$ be two flat tori, defined by unit vectors $\bb$ and $\cc$ with non vanishing coordinates. We can assert that:
\begin{prop}
\label{prop1}
The minimum distance between two points on these flat tori is 
\begin{equation}
d(T_{\cc},T_{\bb})= \left\|  \cc- \bb \right\| = \left( \sum_{i=1}^L (c_i - b_i)^2\right)^{1/2}.
\end{equation}
\end{prop}
\begin{proof}
This follows from observing that
\begin{equation}
D^2:=\|\Phi_{\cc}(\pmb{u})-\Phi_{\cc}(\pmb{v})\|^2 \geq \|\cc-\bb\|^2
\end{equation}
%The proof follows by expressing the square distance 
%$$
%D   =d^{2}(\Phi_{(c_{1},...,c_{L})}(u_{1},u_{2},...,u_{L}),\Phi
%_{(b_{1},...,b_{L})}(v_{1},v_{2},...,v_{L}))
%$$
%and verifying that 
%$$
%D\geq \left\|  (c_{1},...,c_{L})-(b_{1},...,b_{L})\right\|^2,
%$$
with equality holding if, and only if,
$$
\frac{u_{i}}{c_{i}} = \frac{v_{i}}{b_{i}}, \ \ \forall i%
$$
\end{proof}

Note that Proposition \ref{prop1} can also be extended to degenerate tori by replacing $ c_{i} \left(\cos \frac{u_{i}}{c_{i}},\sin \frac{u_{i}}{c_{i}} \right)$ by  $(0,0)$ if $c_i=0$ in  \ref{phi}. 

The distance between two points on the same torus $T_{\pmb{c}}$ given by
\begin{equation*}
d(\Phi_{\pmb{c}}(\pmb{u}),\Phi_{\pmb{c}}(\pmb{v}))=2\sqrt{\sum c_{i}^{2}\sin^{2}(\frac{u_{i}-v_{i}%
}{2c_{i}})}\label{SameTorus}%
\end{equation*}
is bounded in terms of $\|\pmb{u}-\pmb{v}\|$ by the following proposition.
%and we can also establish relations between the distance between two points in $\R^{L}$ and their images by $\Phi_{\pmb{c}}:$
\begin{prop}
\label{prop2}
Let $\pmb{c=}(c_{1},c_{2},..,c_{L})$, $ \left\|  \cc  \right\|  =1$,
and let  $\pmb{u}, \pmb{v} \in \PP_{\cc}$. Let $\Delta =  \left\|\pmb{u}-\pmb{v}\right\|$ and $\delta= \left\|\Phi_{\pmb{c}}(\pmb{u})- \Phi_{\pmb{c}}(\pmb{v}) \right\|$.
Then
$$
\displaystyle 
\frac{2}{\pi}\Delta\leq\dfrac{\sin\frac{\Delta}{2c_{\xi}}}{\frac{\Delta}{2c_{\xi}%
}}\Delta\leq\delta\leq\dfrac{\sin\frac{\Delta}{2}}{\frac{\Delta}{2}}\Delta
\leq\Delta
$$
where $\displaystyle \xi = \arg \min_i (c_i)$.
\\
\end{prop}
\begin{proof} This proposition is an extension
of a result presented in \cite{Cos2003} and its proof follows similar arguments. We can show that for fixed $\Delta$, the mi\-nimum and the
maximum distortion, which correspond to maximum and minimum bending, occur
whenever \linebreak $\pmb{u}- \pmb{v}= \Delta \cc $ and
\ $\pmb{u} - \pmb{v}= \Delta \pmb{e}_{\xi}$ respectively, where $\pmb{e}_j$ denotes the canonical unit basis vector which is nonzero only along the $j$th coordinate.
\end{proof}

\section{Torus Layer Spherical Codes in $\R^{2L}$}
\label{sec:TLSC}

\begin{figure}[h!]
	\centering
		\includegraphics[scale=0.6]{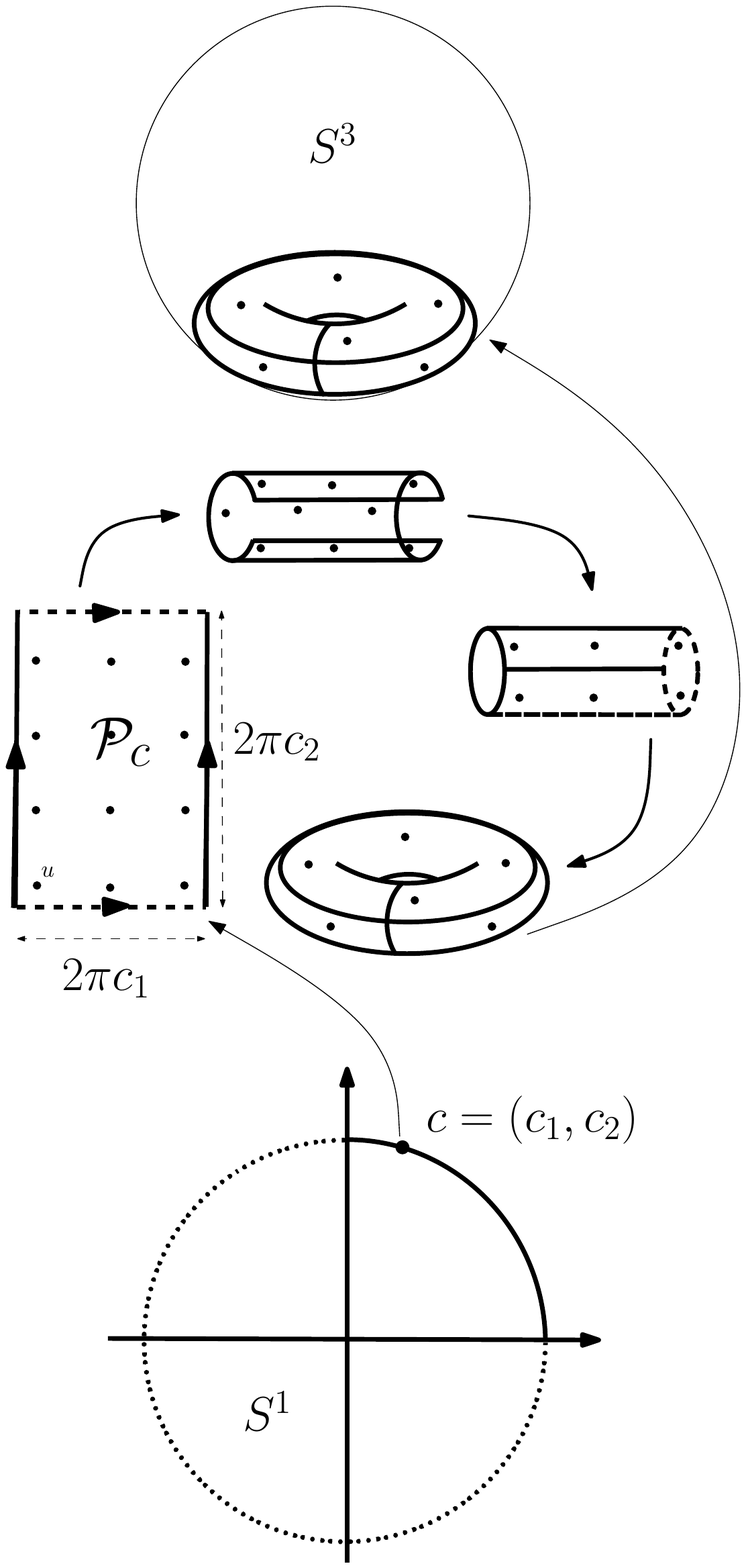}
		\caption{Ilustration of the construction of a four dimensional torus layer spherical code.}
  	 \label{fig:decombp}
\end{figure}

Our goal is to construct a spherical code in $\R^{2L}$ with minimum distance equal to a given value $d$. We denote such a code by $TLSC(2L,d)$.

Before presenting a formal construction technique, we describe the main idea. Given a distance $d \in (0,\sqrt{2}]$, we firstly define a finite set of tori on $S^{2L-1}$ such that the minimum distance, according to Proposition \ref{prop1}, between any two of these tori is greater than $d$. Then, for each one of these tori, a finite set of points is chosen in $\R^L$ such that the distance between any two points, when embedded in $\R^{2L}$ by the standard parametrization (\ref{phi}), is greater than $d$. This set of points may belong to a $L-$dimensional lattice, restricted to the hyperbox (\ref{para}) or to any other suitably chosen set. The $TLSC(2L,d)$ is the union of the images under (\ref{phi}) of each finite sets of points, one for each torus. Figure \ref{fig:decombp} illustrates the construction of a $TLSC(4,d)$.

Note that the pre-image of the points in a single layer of $TLSC(2L,d)$ lie inside an  $L-$dimensional box and hence we are working in half of the code dimension. For not that small values of $d$, our codes compare favorably in terms of code size with previous codes \cite{Gamal}, \cite{Hamkins1} and \cite{Hamkins2} (see Tables \ref{TL5} and \ref{TL4}). In addition, the group structure of our code in each layer allows efficient storage and decoding.

\subsection{The construction of Torus Layer Spherical Codes}

Let $L \geq 2$ and $d \in (0,\sqrt{2}]$. Let $SC(L,d)$ be an $L-$dimensional spherical code with minimum distance greater than $d$. The code $TLSC(2L,d)$ is constructed in two steps as follows:
\begin{itemize}
	\item[(i)] Select the points in the $SC(L,d)$ which have only nonnegative coordinates. This sub-code is denoted by $$SC(L,d)_+ = \left\{c  \in SC(L,d), c_i \geq 0, \ \ 1 \leq i \leq L  \right\}.$$
	Each point $\cc \in SC(L,d)_+$ defines a flat torus $T_{\cc}$ in the unit sphere in $\R^{2L}$ and hence a hyperbox $\PP_{\cc}$, according to ( \ref{para}). 
  \item[(ii)] For each torus $T_{\cc}$, defined by $SC(L,d)_+$, look for the largest set of points  $Y_{T_{\cc}} \subset \PP_{\cc}$ such that $$\left\| \Phi_{\cc}(\pmb{y}) - \Phi_{\cc}(\pmb{x}) \right\| \geq d \ \ \forall \pmb{x}, \pmb{y} \in Y_{T_{\cc}}.$$
\end{itemize}

The performance of a $TLSC$ is directly related to the methods used for constructing $SC(L,d)_+$ and $Y_{T_{\cc}}$. In (i) we must choose a $SC(L,d)_+$ with good density and, if possible, some symmetries. For this purpose we may use any suitable spherical code, e.g., a $L$-dimensional $TLSC$, or some other known structured spherical codes, such as wrapped \cite{Hamkins1}, laminated \cite{Hamkins2} or apple peeling \cite{Gamal} spherical codes. We could also use a non structured spherical code, e.g. one of the  codes listed at \cite{Sloane:Spherical}. Since the cardinality of the set $SC(L,d)_+$ is much smaller than the final code, unstructured spherical codes are also acceptable. For (ii) a good option is to consider lattice points inside each hyperbox $\PP_{\cc}$. Through the maps $\Phi_{\cc}$ they generate group codes in each torus layer \cite{SIQ}. In the next Sections we present examples of this construction.

A \textit{two step-}$TLSC(2L+1,d)$ can also be constructed in odd dimensions by first slicing the unit sphere $S^{2L} \subset \R^{2L+1}$ through  hyperplanes perpendicular to the canonical vector $e_{2L+1}$, such that the minimum distance between two hyperplanes is at least $d$. Table \ref{TL5} shows a comparison between $TLSC(5,d)$ and the apple peeling codes presented in \cite{Gamal} for the same distance.\\

\section{A piecewise cyclic four dimensional TLSC}
\label{TLSC4}

In order to clarify the technique and present an example, we construct a torus layer spherical code $TLSC(4,d)$. For each layer we will design a cyclic group code  so that the resulting code will be a \textit{piecewise cyclic} spherical code.

Step (i) in this construction is to choose a good spherical code in $\R^2$. Since the best spherical code in $\R^2$ with minimum distance $d$ is unique (up to rotation) and is symmetric (the code is the set of vertices of a regular polygon inscribed in the unit circle), the only design choice is to determine  a good rotation for  points in the positive quadrant. 

Our approach is to select points located symmetrically in relation to the line with unit slope.
Thus
\begin{equation}
\label{SC2d}
SC(2,d)_+ = \left\{(\cos(\alpha_{\pm j}), \sin(\alpha_{\pm j})), 0 \leq \alpha_{\pm j} \leq \frac{\pi}{2} \right\},
\end{equation}
where
\begin{eqnarray}
\label{eqalpha}
\lefteqn{ \alpha_{\pm j} = \frac{\pi}{4}  \pm (2j-1)\arcsin{\left( \frac{d}{2} \right),}}   \nonumber  \\
    &  1 \leq\ j \leq \left\lfloor \frac{\pi- 2 \arcsin{(d/2)}}{8 \arcsin{(d/2)}}\right\rfloor .
 \end{eqnarray}
%\begin{figure}[h!]
%	\centering
%		\includegraphics[width=5cm]{nova}
%		\caption{Fig M-psk in $\pi/4$}
%	\label{fig:empty}
%\end{figure}

\begin{figure}[htbf]
	\centering
		\includegraphics[width=7.9cm]{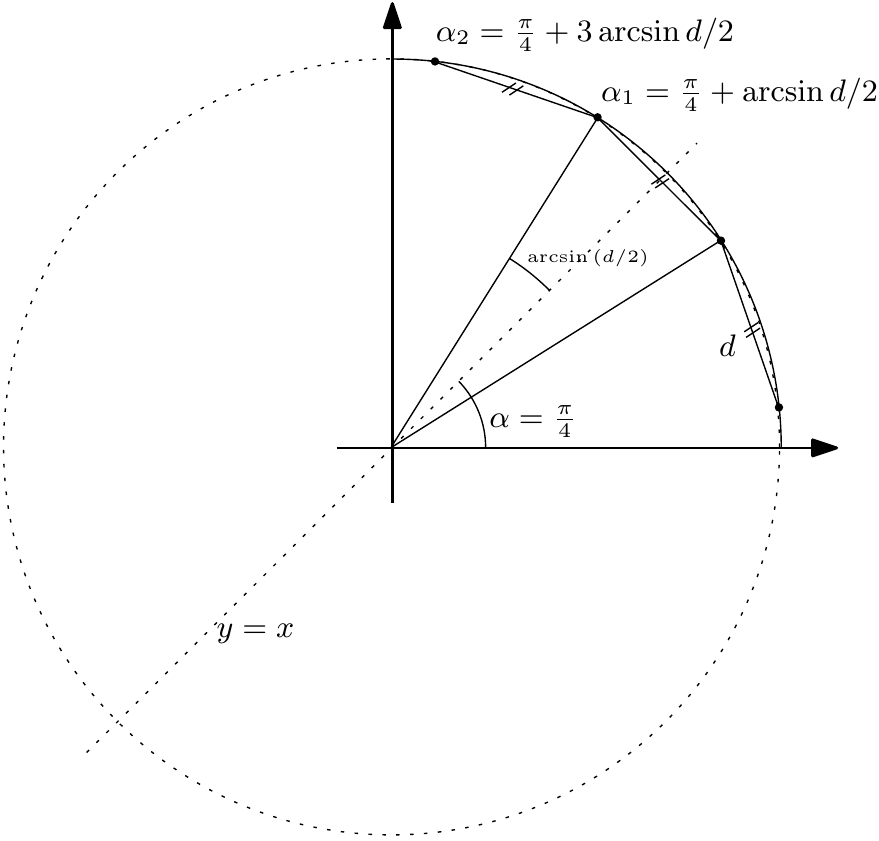}
	\caption{$SC(2,d)_+ $ symmetric in relation to $y=x$  }
	\label{fig:nova.pdf}
\end{figure}
Note that each flattened torus $T_{\alpha_{\pm j}}$ is a rectangle with sides of length $2 \pi \cos(\alpha_{\pm j})$ and $2 \pi \sin(\alpha_{\pm j})$.

In step (ii), we fill each torus $T_{\alpha_{\pm j}}$ searching for the largest cyclic group code, with minimum distance $d$ and initial vector \linebreak $x_{0_{\pm j}} = ( \cos(\alpha_{\pm j}),0,\sin(\alpha_{\pm j}),0)$. 

A cyclic group code of order $M$ in $R^4$ is an orbit of an initial unit vector under action of the cyclic group of orthogonal matrices generated by a  $(2 \times 2)$-block diagonal matrix  \begin{eqnarray*}
\lefteqn{G_j= }  \\ 
&  \left( 
  \begin{array} {cccc} \cos(\frac{2\pi g_{j_1}}{M}) &  \sin(\frac{2\pi g_{j_1}}{M}) & 0 & 0 \\
   -\sin(\frac{2\pi g_{j_1}}{M}) & \cos(\frac{2\pi g_{j_1}}{M})                     & 0 & 0 \\
    0 & 0 & \cos(\frac{2\pi g_{j_2}}{M}) &  \sin(\frac{2\pi g_{j_2}}{M}) \\
    0 & 0 & -\sin(\frac{2\pi g_{j_2}}{M}) & \cos(\frac{2\pi g_{j_2}}{M})
  \end{array} \right),
\end{eqnarray*}
where $\gcd(g_{j_1},g_{j_2})=1$. The pre-image of this cyclic group code by $\Phi_{\cc}$ is a lattice in $\R^2$ \cite{SIQ}.

The search for the largest group code for each torus  $T_{\alpha_j}$, attained by an $L$-dimensional lattice can be accomplished based on \cite{Torezzan} and a simplified algorithm is described here.
\begin{algorithm}[htbf]
%\SetLine
\KwIn{$d$, $\alpha$}
\KwOut{Generators: $\{g_{j_1},g_{j_2}\}$}

$x_{0} = ( \cos(\alpha),0,\sin(\alpha),0)$;\\

$\displaystyle M = \left\lfloor  \frac{\pi^2 \cos(\alpha) \sin(\alpha) }{2 \sqrt{3} \arcsin{\left(\frac{d}{4}\right)}^2} \right\rfloor$;\\

$continue = 1;$

\While{continue}{
\For{$g_{j_1} = 1$ to $\left\lfloor \frac{M}{2}\right\rfloor$}{
\For{$g_{j_2} = 1$ to $\left\lfloor \frac{M}{2}\right\rfloor$}{
\If{$\gcd(g_{j_1},g_{j_2})=1$}{
$\displaystyle \bar{d} =  \min_{{\tiny 1 \leq i \leq \left\lfloor \frac{M}{2}\right\rfloor}} \left\|(G_j)^i x_0 - x_0\right\|$;

\If{$\bar{d} \geq d$}{
Return $\{g_{j_1},g_{j_2}\}$;

$continue = 0;$

\textbf{Stop};
}
}
}
}
$M = M - 1$;
}
\caption{Algorithm to search for the best cyclic group code in  $R^4$ for a given initial vector $x_{0_j}$.}
\label{alg:mine}
\end{algorithm}

This problem is, loosely, a dual to the initial vector problem (IVP) \cite{Slepian}, \cite{BiglieriElia}, which is a classic problem  in group codes. In IVP is given a group and the problem is to find a unit vector in order to maximize the minimum distance in its orbit. Here we have the initial vector $x_{0_{\pm j}}$ (which defines the torus $T_{\alpha_{\pm j}}$) and wish to find the largest cyclic group such that the minimum distance in its orbit is, at least, a previously fixed value $d$.

\subsection{An example: $TLSC(4,0.3)$}
We now illustrate the construction of a quasi-cyclic torus layer spherical code for $n = 4$ and $d=0.3$.

\begin{itemize}
	\item From \ref{eqalpha} we get $ \alpha_{+1} = 0.935966$, $ \alpha_{+2} = 1.2371$, and \linebreak $ \alpha_{+3} = 1.53824$, which define the points of $SC(L,d)_+$ above the line $y=x$, according (\ref{SC2d}).
	\item For each torus {$T_{\alpha_{+j}}$}, we have found the largest cyclic group code using the algorithm \ref{alg:mine}. The result is:\\
	
\begin{table}[h!]
	\footnotesize
	\centering
\begin{tabular}{|c|c|c|c|c|c|}
 \text{$\alpha$} & $\cos(\alpha)$ & $\sin(\alpha)$ & \text{dmin} & M & $(g_{i1},g_{i2})$ \\
 0.935966  & 0.593041 & 0.805173 & 0.30225 	& 233 	& \{1,98\} \\
 1.237103  & 0.327535 & 0.944839 & 0.301406 & 146 & \{22,1\} \\
 1.538240  & 0.032551 & 0.99947  & 0.312869 & 20 	& \{0,1\}
\end{tabular}
 	\caption{Part 1 of $TLSC(4,0.3)$: tori above the slope line.}
	\label{tab_cod1}
\end{table}
	\item Finally, for each torus {$T_{\alpha_{+j}}$} we consider the symmetric layer {$T_{\alpha_{-j}}$}, just interchanging the coordinates.
\begin{table}[htb]
	\footnotesize
	\centering
\begin{tabular}{|c|c|c|c|c|c|}
 \text{$\alpha$} & $\sin(\alpha)$ & $\cos(\alpha)$ & \text{dmin} & M & \text{Generator} \\
 0.634829  & 0.805173 & 0.593041 & 0.30225  &  233 & \{98,1\} \\
 0.333694  & 0.944839 & 0.327535 & 0.301406  & 146 & \{1,22\} \\
 0.032559  & 0.99947  & 0.032551 & 0.312869   & 20 & \{1,0\}
\end{tabular}
 	\caption{Part 2 of $TLSC(4,0.3)$: tori below the slope line.}
	\label{tab_cod2}
\end{table}\\
\end{itemize}

The resulting code $TLSC(4,0.3)$ has 6 layers, pairwise symmetric, with 
$20,146,233,233,146,20$ points respectively and thus the entire code has $798$ points. 

Due to the symmetry and group structure of this code, in order to store all $798$ codewords in this $TLSC(4,0.3)$ is only required columns $1, 5$ and $6$ of table \ref{tab_cod1}. The constructiveness of the codewords is a good aspect of the Torus Layers Spherical Codes when compared with several other known construction of spherical codes.

In the Table \ref{TL4}, we compare torus layer spherical codes to three other known spherical codes: apple-peeling \cite{Gamal}, wrapped \cite{Hamkins1} and laminated  \cite{Hamkins2}, at various minimum distances $d$. 

\begin{table}[htb]
	\small
	\centering
			\begin{tabular}{|c|c|c|c|c|}
 \hline d & \text{TLSC(4,d)} & \text{apple}-\text{peeling} & \text{wrapped} & \text{laminated}\\
\hline
  \hline 0.5   & 172 & 136  &  * &  *  \\
  \hline 0.4   & 308 & 268 & *  &   * \\
  \hline 0.3   & 798 & 676  & *  & *  \\
  \hline 0.2   & 2,718 & 2,348  & *  & * \\
  \hline 0.1   & 22,406 & 19,364 & 17,198 & 16,976 \\
  \hline 0.01  & 2.27 $\times 10^{7}$ & 1.97 $\times 10^{7}$ & 2.31 $\times 10^{7}$ & 2.31 $\times 10^{7}$ \\
 \hline
		\end{tabular}
 	\caption{Four-dimensional code sizes at various minimum distances. (*): unknown values.}
	\label{TL4}
\end{table}

Using four dimensional codes and successive slices of the $S^{4}\subset \R^{5}$ by hiperplanes we constructed some TLSC(5,d). Although $S^4$ is not foliated by flat tori, the codes constructed in layers of tori in $S^3$ and lifted to $S^4$ exhibit good performance as illustred in Table \ref{TL5}.

\begin{table}[htb]
	\small
	\centering
			\begin{tabular}{|c|c|c|}
			\hline \text {d} &\text {TLSC(5,d)} &\text {APC(5,d)} \\
			 \hline 0.8 & 48 & 48 \\
			 \hline 0.7 & 98 & 64 \\ 
			 \hline 0.6 & 196 & 160 \\ 
			 \hline 0.5 & 374 & 336 \\ 
			 \hline 0.4 & 872 & 872 \\ 
			 \hline 0.3 & 3,232 & 2,960 \\ 
			 \hline 0.2 & 17,140 & 15,424 \\ 
			 \hline 0.1 & 296,426  & 256,760 \\ 
			 \hline 0.05& 4,824,018  & 4,164,152 \\
			 \hline
		\end{tabular}
 	\caption{Comparison between $5-$dimensional torus layer and apple peeling spherical codes at various minimum distances}
	\label{TL5}
\end{table}

\section{Orthogonal sublattices and piecewise commutative group codes}
\label{group}

In this section we discuss how dense lattices in $R^L$ with good packing density can be used to construct a well structured TLSC in $R^{2L}$. This is done by considering orthogonal $L$-dimensional sublattices in step (ii) of our construction. 

We first complete the step (i) of our construction for a given minimum distance $d$ where each point $\cc \in SC(d,L)_+ \subset S^{L-1}$ defines precisely the lengths of an orthogonal hyperbox \linebreak $\PP_{\cc} \subset \R^L$. For a chosen $L$-dimensional lattice $\Lambda$, we look for an orthogonal sublattice $\Lambda_1$ such that the fundamental region $F$ of $\Lambda_1$ approaches $\PP_{\cc}$. Then $\Lambda_1$ should be scaled to a lattice $\tilde{\Lambda_1}$ in order to get $\PP_{\cc}$ as its fundamental region.

The lattice points of $\tilde{\Lambda_1}$ inside $\PP_{\cc}$ are identified with the quotient $\tilde{\Lambda}/\tilde{\Lambda_1}$ and inherit the associated group structure, since $\tilde{\Lambda}/\tilde{\Lambda_1}$ is isomorphic to $\Lambda/\Lambda_1$. The image by $\Phi_{\cc}$ of these  points defines a commutative group code in $\R^{2L}$, with initial vector defined by vector $\cc \in SC(d,L)_+$ \cite{SIQ, Torezzan}. 

In addition, let $B$ and $B_1$ be the generator matrices of the lattices $\tilde{\Lambda}$ and $\tilde{\Lambda_1}$ respectively and let $Q$ be an integer matrix such that $B_1= B Q$. Then the characterization and the set of generators of the group $\tilde{\Lambda}/\tilde{\Lambda_1}$ can be obtained from the standard Smith Normal Form of $Q$ \cite{Cohen, Torezzan}.

This strategy can be recursively applied to all  flat torus layers defined by $SC(d,L)_+$ to get a piecewise commutative TLSC in $R^{2L}$, i.e. a spherical code constituted of layers of commutative group codes.

\subsection{A $48$-dimensional TLSC from the Leech Lattice}

To illustrate the construction described above, we present next a piecewise commutative group code in $R^{48}$ with minimum distance $d=0.1$, designed from an orthogonal sublattice of the Leech Lattice. Although the number of points in this code is of order $10^{34}$ the construction is quite simple and do not need the storage of the points. In addition, there is a fashion labeling for all points in this code induced by the commutative group code in each layer.

In order to simplify the step (i) and focus our attention in step (ii), we start from a set of $24$ points in the unit sphere $S^{23}$ consisting of all permutations of the vector $$\displaystyle \cc(t) = \frac{ (t, 1, 1, \cdots, 1)}{\sqrt{23+t^2}},$$
where $t(d)>0$ is chosen in order to guarantee the desired minimum distance $d$ between any two permutation.

In this example, to design a code with minimum distance $d=0.1$, we get $t(0.1) \cong 1.35234$ and

$$
\cc(1.35234) \cong  (0.271399, 0.200688, \cdots, 0.200688) \in S^{23}
$$
Each permutation of $\cc(t)$ defines a flat torus on the surface of $S^{47} \subset \R^{48}$ which can be flattened into a $24$-dimensional hyperbox $\PP_{\cc}$. Therefore, the step (i) of our construction is done. 

The next step is to find a discrete set of points inside each hiperbox defined above. Since the $24$ hiperbox differs only by rotations (or interchange of coordinates) it is enough to solve this problem for one of those tori and we show next how to use the Leech lattice to solve that.

Consider a standard Leech Lattice rescaled by a factor \linebreak $\beta(d) = 0.10187$ to assure minimum distance $d=0.1$ in $R^{48}$, after apply the function $\Phi_{\cc(t)}$ (according Proposition  \ref{prop2}). 

Let $B$ be the generator matrix of the scaled Leech lattice $\Lambda_{24\beta}$ with minimum distance $\beta$. Since the Leech lattice contains the sublattice $4\Z^{24}$, we can determine a factor $\alpha$ such that the lattice generated by  $\alpha I$, where $I$ is the identity matrix of order $24$, is an orthogonal sublattice of $\Lambda_{24\beta}$. 

Now, let $\vv(t)$ be an integer vector defined by 
$$  \vv_i=\left\lfloor \frac{2 \pi \cc(t)_i - \beta}{\alpha} \right\rfloor.$$
Since the length of each edge of the flattened hiperbox $\PP_{\cc}$ is given by $2\pi \cc_i(t)$, each coordinate $i$ of vector $\vv$ determines the maximal number of times vector $\alpha e_i$ can be placed in each canonical direction of $\PP_{\cc}$, i.e. vector $v$ defines the largest orthogonal sublattice of $\Lambda_{24\beta}$ in $\PP_{\cc}$. In our example
$$
\vv(1.35234) = (11,8,8,\cdots,8).
$$

So the matrix $B_1 = \alpha \mbox{diag}(\vv)$ is a generator matrix of a sublattice $\displaystyle \tilde{\Lambda}_{24\beta} \subset \Lambda_{24\beta}$ which fundamental region approaches hiperbox $\PP_{\cc}$. Therefore the number $M_i(t)$ of lattice points in this hiperbox is given by the volume of the quotient between these lattices, i.e.
$$
M_i(t) = \dfrac{det(B_1)}{det(B)}.
$$
In this example $M_i(1.35234) = 4.46213 \times 10^{32}$. Since all $24$ torus layers are symmetric, the total number of point in the $48$-dimensional spherical code is given by
$$
M = 24 \times M_i(t) = 1.07091 \times 10^{34}.
$$

Since the points inside the hiperbox $\PP_{\cc}$ are defined by a quotient of Abelian groups, we can use the standard Smith Normal Form \cite{Cohen} of the matrix $Q = B^{-1} B_1$ to classify this group and to get the set of generators. 

In this example, in each layer the commutative group is isomorphic to $\mathbb{Z}_{64 }\times \mathbb{Z}_{256}\times \mathbb{Z}_{512}{}^9\times \mathbb{Z}_{11264}.$

In addition, every point $\xx_i \in TLSC(48,0.1)$, \linebreak $1 \leq i \leq 4.46213 \times 10^{32}$ in each layer of the spherical code can be generated as the orbit though a product of power of rotation matrices by a initial vector $\cc(t)$ as follows
$$
\xx_i = \left( W_1^{k_1} \, . \,   W_2^{k_2} \, . \,  W_{j}^{k_{j}} \, . \,  W_{10}^{k_{10}} \right) \, . \,  \xo,
$$
where
$$
\xo = (\cc(t)_1,0,\cc(t)_2,0, \cdots, \cc(t)_{24},0) \in S^{47} \subset \R^{48}, 
$$
$$
\begin{array}{ccccc}
0 & \leq & k_1 & \leq & 63\\
0 & \leq & k_2 & \leq & 255\\
0 & \leq & k_{j} & \leq & 511\\
0 & \leq & k_{10} & \leq & 11263\\
1 & \leq & j & \leq & 9 
\end{array}
$$
and  $W_1, W_2, W_j, W_{10}$ represent, respectively, the generators of subgroups of rotation matrix in $\mathcal{O}(48)$  (orthogonal $48 \times 48$ matrices) isomorphic to \linebreak $$\mathbb{Z}_{64 }, \mathbb{Z}_{256}, \mathbb{Z}_{512}{}^9, \mathbb{Z}_{11264}.$$

Since the group structure are the same for all $24$ layers in this $TLSC(48,0.1)$ spherical code, there is a natural labeling for all the $1.07091 \times 10^{34}$ points in this code induced by the set of permutation vectors $\Sc(L,d)_+$ and the commutative group code in each layer. It means that we are able to generated each one of the code points independently, which is a very useful property in many applications, specially channel coding and vector quantization.

\section{Bounds and density of TLSC}

\label{sec:Bounds}

\subsubsection{The grid TLSC}

In this section we derive a lower and an upper bound for the number of points in a torus layer spherical code. Both bounds depend on a code $SC(L,d)$ in $L$-dimensions. More specifically, to present the bounds we will assume that we have completed step (i) of the construction, i.e., assume we have selected and stored $k$ points in $SC(L,d)_+$.

For given $d$, we construct a  $TLSC$, by choosing $Y_{T_{\cc}}$ as a subset of  the rectangular lattice which lies in the hyperbox $\PP_{\cc}$ (\ref{para}). Let $\cc_i=(c_{i1, c_{i2}, \cdots, c_{iL}}) \in  SC(L,d)_+$ and $\uu$ be a point in the rectangular lattice is given by $\uu = \sum_{j=1}^L m_j a_{ij} \pmb{e}_j$, where $a_{ij}$ is the increment along the $j$th coordinate, $\pmb{e}_j$ is the $j$th unit canonical basis vector and $m_j$ is an integer, $j=1,2,\ldots,L$.
\begin{figure}[h!]
	\centering
		\includegraphics[width=8cm]{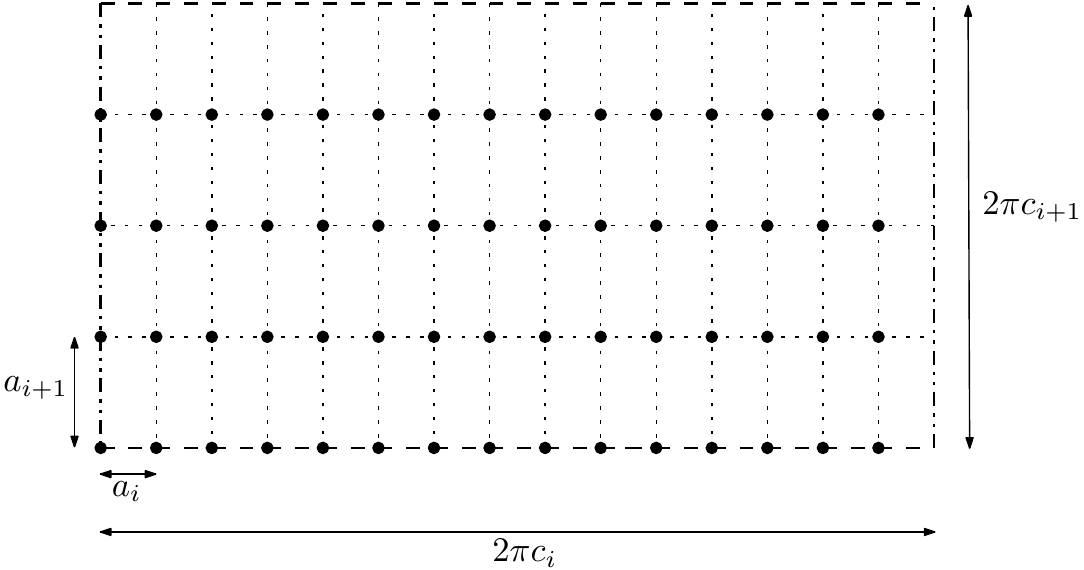}
		\caption{Illustration of a grid in a torus.}
	\label{fig:emp2}
\end{figure}
Since the hyperbox $\PP_{\cc_{i}}$ has length $2 \pi c_{ij}$ along the $j$th coordinate, we can determine the maximum number of lattice points in the hyperbox $\PP_{\cc_{i}}$ by finding the smallest $a_{ij} \in \R$ which satisfies
$$
d (\Phi_{\cc}(a_{ij}\pmb{e}_j), \Phi_{\cc}(0)) = 2 c_{ij}  \sin{\left(\frac{a_{ij}}{2c_{ij}} \right)} \geq d.
$$
This leads to the following lower bound on the number of points in the code.
%Therefore we can put $\left\lfloor \frac{\pi}{ \arcsin{\left(\frac{d}{2c_i}\right)}} \right\rfloor$ lines parallel to this edge.
%
%In doing so, for each edge of the flattened torus $T_{\cc}$, we can get a subset of the ``rectangular'' lattice, leading to the following preposition. 
\begin{prop}
For a given $SC(L,d)_+$ with $k$ points, we can design a $TLSC(2L,d)$ with $M(2L,d)$ points where 
 $$\displaystyle M(2L,d) = \sum^k_{i=1} \prod_{j=1}^L W_{ij}$$
and
$$
\displaystyle
W_{ij} = \left\{ 
\begin{array}{cc}
\left\lfloor \dfrac{\pi}{\arcsin{\dfrac{d}{2 c_{ij}}}} \right\rfloor, & \mbox{ if } \left|\dfrac{d}{2 c_{ij}}\right| \leqslant 1 \\
	1, & \mbox{ if } \left|\dfrac{d}{2 c_{ij}}\right| > 1 
\end{array}
\right.
$$
\end{prop}

%\begin{eqnarray}
%\max \left\{ \left\lfloor \frac{\pi}{\arcsin{\frac{d}{2 \delta_{i,j}}}} \right\rfloor, 1 \right\}, & 
%\text{if } {\frac{d}{2 \delta_{i,j}}} < 1
%\end{eqnarray}

\subsubsection{An upper bound}

Suppose again that we have $k$ tori defined by a pre-selected $SC(L,d)_+$. Let $\cc_i=(c_{i1},c_{i2},\cdots,c_{ik})$ be the $i$th radius vector in $SC(L,d)_+$. Without loss of ge\-nerality we consider $c_{ij} \geq c_{ij+1} \forall i,j$. We can use the bound presented in \cite{SIQ}[Proposition 7] to obtain the following upper bound for the number of points $M_{T_{\cc_i}}$ in each torus $T_{\cc_i}$:
$$\displaystyle M_{T_{c_i}} \leq \left\lfloor  \frac{\pi^L}{(\arcsin{\frac{d}{4}})^L} \prod_{j=1}^L c_{ij} \Lambda_{L} \right\rfloor$$ where $\Lambda_{L}$ is
the maximum center density of a packing in $\R^{L}$.\\
However, there is an additional consideration. In case of the smallest coordinate of $\cc_i$ is small enough we may get
$$ \left\lfloor  \frac{\pi^L}{(\arcsin{\frac{d}{4}})^L} \prod_{j=1}^L (c_{ij}) \Lambda_{L} \right\rfloor = 0 .$$ 
In this case, we should remove the smallest coordinate of $\cc_i$, to project the torus $T_{\cc_j}$ into one less dimension and consider a code in that dimension. This process of projection should be carried out until a non zero value for $M_{T_i}$ is found. This is equivalent to place the points in a face of the hiperbox $\PP_{\cc_i}$. We remark that, in the worst case all but one $c_{ij}$ is zero and the torus degenerates to a circle where at least $$\max \left\{ \left\lfloor \frac{\pi}{\arcsin{\frac{d}{2 c_{i1}}}} \right\rfloor, 1 \right\}$$ points can be placed.

Let $M_{T_{c_i}}^p$, for some $1 \leq p \leq L$, be the number of point that fit in the $p$-face of the hiperbox $\PP_{\cc_i}$, ($p=1$ corresponding to the one-dimensional degenerated torus obtained when just the first coordinate $\cc_{i1}$ is non zero an $p=L$ corresponding to the $L$-dimensional flat torus obtained when all $\cc_{ij}$ are non zero. Thus
$$
M_{T_{c_i}}^p =  \left\lfloor  \frac{\pi^p}{(\arcsin{\frac{d}{4}})^p} \prod_{j=1}^p (c_{ij}) \Lambda_{p} \right\rfloor.
$$

Therefore the maximum number of points in each tori is given by
$$
M_{T_{c_i}}^* =\displaystyle \max_{1 \leq p \leq k} M_{T_{c_i}}^p.
$$
This allow us to dirive an upper bound for the number of points $M(2L,d)$ of a $TLSC(2L,d)$.

\begin{prop}
Given a $SC(L,d)_+$ with $k$ points, the number of points in a $TLSC(2L,d)$ satisfies
$$
M(2k,d) \leq \sum_{i=1}^{k}   M^*_{T_{c_i}}  
$$
\end{prop}

Table \ref{tab:adf} shows a comparison between these bounds and a $TLSC(4,d)$ designed in the Section \ref{TLSC4}. Note the tightness of the upper bound when the distance decreases.

\begin{table}[htb]
	\small
	\centering
			\begin{tabular}{|c|c|c|c|}
 \hline d & \text{TLSC(4,d)}  & \text{grid lower bound} & \text{upper bound}\\
\hline
  \hline 0.5   & 172 & 120 & 194  \\
  \hline 0.4   & 308 & 208 & 360 \\
  \hline 0.3   & 798 & 612 & 826 \\
  \hline 0.2   & 2718 & 2148& 2854\\
  \hline 0.1   & 22,406 & 18,884 & 22,478\\
  \hline 0.01  & 2.279 $\times 10^{7}$ & 1.967 $\times 10^{7}$ & 2.279$\times 10^{7}$ \\
 \hline
		\end{tabular}
 	\caption{Bounds for 4-dimensional torus layer spherical codes at various minimum distances}
	\label{tab:adf}
\end{table}

%\subsection{TLSC in odd dimensions}

\subsection{Density of TLSC}

In this section we analyze the density of the torus layer spherical codes. 

Let $\Gamma$ be the standard Gamma Function, 
$$
\Gamma(x) = \int_0^{\infty} e^{-t} t^{x-1} \, \mathrm{d}t.
$$
We shall use 
$$\Su_L := \frac{L \pi^{L/2}}{\Gamma((L/2)+1)}$$
 for the $(L-1)$-dimensional volume (surface area) of the sphere $S^{L-1} \subset R^L$ and 
 $$\Vo_L := \frac{\pi^{L/2}} {\Gamma((L/2)+1)}$$
  for the $L$-dimensional volume of the ball bounded by $S^{L-1}$. 
  
We also use $$\Sc(\theta/2, L):= \Su_{L-1} \int_0^{\theta/2}\sin^{L-2}x dx$$
   for the $(L-1)$-dimensional volume of a spherical cap on the sphere $S^{2L-1}$ with angular radius $\theta = 2 \arcsin(d/2)$. By using the Taylor series of $\sin(x)$ and some standard calculations (see e.g. \cite{Hamkins1}) we can obtain  
   $$\Sc(\theta/2, L)=\Vo_{L-1} \left(\frac{d}{2}\right)^{L-1} + O(d^{L+1}).$$

The density of a $L$-dimensional spherical code with minimum distance $d$ and $M$ codewords is proportion of the area of $S^{L-1}$ occupied by the union of the spherical caps centered at the codewords and with angular radius $\theta = 2 \arcsin(d/2)$, that is,
$$
\Delta_{SC} = \frac{ \Sc(d/2, L) \ \ M }{\Su_{L}}.
$$

For a given minimum distance $d$, the maximum cardinality of a $L$-dimensional spherical code is unknown for all $L \geq 3$, except for a handful of values of $d$ \cite{zino}, therefore the problem of determine the maximum density of a $L$-dimensional spherical code is approached through bounds. 

Next proposition approaches the density of a torus layer spherical code, for asymptotically small $d$.

In what follow we denote $f(d) \simeq g(d)$ if $$\displaystyle \lim_{d\rightarrow 0} \frac{f(d)}{g(d)} =1.$$

\begin{prop}
\label{prop_dens}
The torus layer spherical code density $\Delta_{TLSC}$ is upper bounded and asymptotically approach the density of $\Delta_{\Lambda_L \times \Lambda_{L-1}}$, where $\Lambda_n$ is the densest lattice in $R^n$.
\end{prop}

\begin{proof} 

The torus layer spherical code density is given by
\begin{equation}
\label{eq:tlsc}
	\Delta_{TLSC} = \frac{ \Sc(d/2, 2L) \ \ M }{\Su_{2L}},
\end{equation}
where $\displaystyle M = \sum_{i=1}^k {M_i}$ is the total number of codewords, $M_i$ is the number of codewords in the $i$-th torus layer and $k$ is the total number of layers on which the code lays on.

When the distance become small, the number of points in each layer can be approached by considering the best $L$-dimensional lattice packing \cite{SIQ}
$$\displaystyle M_i \simeq \frac{\Delta_{\Lambda_L} \prod_{j=1}^L (2\pi c_{ij}) }{(d/2)^L\Vo_L},$$
and therefore
$$
\displaystyle M  \simeq \sum_{i=1}^k { \frac{\prod_{j=1}^L 2\pi c_{ij} \Delta_{\Lambda_L}}{(d/2)^L\Vo_L}}
$$
$$
\displaystyle M  \simeq \frac{\Delta_{\Lambda_L}}{(d/2)^L\Vo_L} \sum_{i=1}^k { \prod_{j=1}^L 2\pi c_{ij}}
$$
Since the sphere $S^{2L-1}$ can be foliated by flat tori, we may assert
$$
\Su_{2L} \simeq \sum_{i=1}^k \prod_{j=1}^L (2\pi c_{ij}) \textit{d}\mathcal{V},
$$
where the element of $(L-1)$-volume \textit{d}$\mathcal{V}$ is the volume of the ``positive'' part of the sphere $S^{L-1}$ divided by the number of tori,
$$
\textit{d}\mathcal{V} \simeq \frac{\Su_{L}}{2^L \left( \frac{\Su_{L} ( \Delta_{\Lambda_{L-1}} ) }{2^L \Sc(d/2,L-1)} \right)} = \frac{\Sc(d/2,L-1)}{\Delta_{\Lambda_{L-1}}}.
$$
Therefore we may assert
$$
\sum_{i=1}^k \prod_{j=1}^L (2\pi c_{ij}) \simeq \frac{\Su_{2L}}{\textit{d}\mathcal{V}},
$$
and the number of codewords can be estimated by 
$$
M \simeq \frac{\Delta_{\Lambda_L}}{(d/2)^L\Vo_L} \frac{\Su_{2L}}{\frac{\Sc(d/2,L-1)}{\Delta_{\Lambda_{L-1}}}} = \frac{\Delta_{\Lambda_L} \Delta_{\Lambda_{L-1}} \Su_{2L}}{(d/2)^L \Vo_L \Sc(d/2,L-1)} =
$$
$$
= \frac{ \frac{V_L}{\det{\Lambda_L}} \frac{V_{L-1}}{\det{\Lambda_{L-1}}} 2L \Vo_{2L} }{(d/2)^L \Vo_L (d/2)^{L-1} V_{L-1}}  = \frac{2L \Vo_{2L}}{(d/2)^{2L-1} \det{\Lambda_{L}} \det{\Lambda_{L-1}} }
$$
Thus, from (\ref{eq:tlsc}), when $d \rightarrow 0$ we get
$$
\Delta_{TLSC} \simeq \frac{ \Sc(d/2, 2L) \frac{2L \Vo_{2L}}{(d/2)^{2L-1} \det{\Lambda_{L}} \det{\Lambda_{L-1}} } } {\Su_{2L}}
$$
and so
$$
\Delta_{TLSC} \simeq \frac{\Vo_{2L-1}}{\det{\Lambda_{L}} \det{\Lambda_{L-1}}},
$$
which is the density of the Cartesian lattice $\Lambda_L \times \Lambda_{L-1}$.
\end{proof} 

It should be remarked here that this asymptotic density is much better than the asymptotic density of apple peeling  \cite{Gamal} construction but certainly worst than the best lattice packing density in $R^{2L-1}$, which can be achieved by the wrapped  \cite{Hamkins1} and laminated  \cite{Hamkins2} codes. On the other hand, for not that small $d$, as we have seen in Sec. \ref{TLSC4} a TLSC can outperform these previous constructions besides having the mentioned features inherit from its group structure.

\section{Decoding}
\label{sec:Decoding}

Given an arbitrary $x \in \R^{n}$ and a $n$-dimensional spherical code $\Sc$, the maximum-likelihood decoding problem is to find
\begin{equation}
	\displaystyle y = \arg \min_{y_i \in \Sc} ||x - y_i||.
	\label{decod}
\end{equation}

For any $x \in \R^n$ and a spherical code $\Sc$
$$
\arg \min_{y \in \Sc} \left\| x - y \right\| = \arg \min_{y \in \Sc} \left\| \dfrac{x}{||x||} - y \right\| .
$$

In fact, let $y \in \Sc$, such that  $$\left\| \dfrac{x}{||x||} - y \right\| \leq \left\| \dfrac{x}{||x||} - z \right\| \ \ \forall \ \ z \in \Sc.$$
Thus,
$$2 - 2 \left\langle  \dfrac{x}{||x||}, y   \right\rangle \leq 2 - 2 \left\langle  \dfrac{x}{||x||}, z   \right\rangle \Rightarrow \left\langle x, y\right\rangle \geq \left\langle x, z\right\rangle,$$
and
$$\left\| x - y \right\| = ||x||^2 +1 -2 \left\langle x,y \right\rangle  \leq ||x||^2 + 1 - 2 \left\langle x,z\right\rangle = \left\| x - z \right\|.$$

Therefore, for decoding any received vector $x$, using a spherical code, we can assume $x$ is a unit vector.

One approach to solve (\ref{decod}) is computing the all the inner products between $x$ and $y_i \in \Sc$ and search for 
$$\displaystyle y = \arg \max_{y_i \in \Sc}\left\langle x, y_i\right\rangle.$$

This process requires $O(Mn+M)$ flops, but the main problem here is that this approach requires the storage of all codebook in the decoder, which is a restrictive requirement for many applications with limited memory. In addition, mostly computations required in the decoding process are done in the half of the code dimension. 

In what follows we address on decoding in a $TLSC$.

For any $2L$-dimensional unit vector $x$ we can write 
{\small
\begin{eqnarray}
x & = &  \left( \gamma_1 \left( \frac{x_1}{\gamma_1}, \frac{x_2}{\gamma_1} \right), \hdots , \gamma_L \left( \frac{x_{2L-1}}{\gamma_L}, \frac{x_{2L}}{\gamma_L} \right)
 \right) \nonumber \\
%x & = &  \left( \sqrt{x_1^2+x_2^2}\left(\frac{x_1}{\sqrt{x_1^2+x_2^2}}, \frac{x_2}{\sqrt{x_1^2+x_2^2}} \right), \hdots ,
% \right) \\
 %x & = &  \left( \sqrt{x_1^2+x_2^2}\left(\frac{x_1}{\sqrt{x_1^2+x_2^2}}, \frac{x_2}{\sqrt{x_1^2+x_2^2}} \right), \hdots ,
%\sqrt{x_{2k-1}^2+x_{k}^2}\left(\frac{x_{2k-1}}{\sqrt{x_{2k-1}^2+x_k^2}}, \frac{x_k}{\sqrt{x_{2k-1}^2+x_k^2}} \right)
% \right)\nonumber \\
x & = &  \left( \gamma_1\left(\cos{\frac{\theta_{1}}{\gamma_1}}, \sin{\frac{\theta_{1}}{\gamma_1}} \right), \hdots ,
\gamma_L \left(\cos{\frac{\theta_{L}}{\gamma_L}},\sin{ \frac{\theta_{L}}{\gamma_L}} \right) \right). \nonumber 
\end{eqnarray}
}
Where,
 \begin{eqnarray}
 \gamma_i   & = & \sqrt{x_{2i-1}^2+x_{2i}^2}, \ \ 1 \leq i \leq L \nonumber \\
 \theta_{i} & = & \arccos{ \left( \dfrac{x_{2i-1}}{\gamma_i} \right) } \gamma_i,  \ \ 1 \leq i \leq L \nonumber.
 \end{eqnarray}
This means that $x$ belongs to a flat torus of radius \linebreak $c_{x}=(\gamma_1,\gamma_2, \cdots, \gamma_L).$ In general, $c_x \notin \Sc(L,d)_+$, i. e. $c_x$ does not defines a layer in the spherical code and we must project $x$ in the closest layer. 

This process involves a spherical decoding in $L$-dimension, considering just the points in $\Sc(L,d)_+$ which defines the layers of tori in $TLSC$. As the number of tori is, in general, much smaller of code's cardinality, it does not increase the complexity of the entire process.

For any $c_i = (c_{i1}, c_{i2}, \cdots, c_{iL}) \in \Sc(L,d)_+$ the vector
$$\bar{x}_i = \left( c_{i1}\left(\cos{\frac{\theta_{1}}{\gamma_1}}, \sin{\frac{\theta_{1}}{\gamma_1}} \right), \hdots ,
c_{iL} \left(\cos{\frac{\theta_{L}}{\gamma_L}},\sin{ \frac{\theta_{L}}{\gamma_L}} \right) \right)$$ is the projection of $x$ in the torus $T_{c_i}$, i.e., 
$$
||x-\bar{x}_i|| \leq || x - y || \,  \forall \, y \in T_{c_i}.
$$
Let $T_{c_\xi}$ be the closest torus to $x$. With high probability, the solution of (\ref{decod}) belongs to the torus $T_{c_{\xi}}$, and can be found by decoding the vector 
$$
z_{\xi} = \psi^{-1}_{c_{\xi}}(\bar{x}_{\xi}) = \left( \frac{\theta_1 c_{\xi_1}}{\gamma_1} ,  \frac{\theta_2 c_{\xi_2}}{\gamma_2}, \cdots, \frac{\theta_L c_{\xi_k}}{\gamma_L} \right)
$$
in the $L$-dimensional hyperbox $\PP_{c_{\xi}}$ using an efficient algorithm in the half of the code's dimension, which depends on the structure of the points in $\PP_{c_{\xi}}$. For instance, for codes designed in previous section, the decoding in $\PP_{c_{\xi}}$ requires $O(L)$ flops and does not need to store the codebook \cite{ConwaySloane}.

% and $\Delta_i = ||x-\bar{x}_i||$ be the smallest distance between $x$ and the torus $T_{c_i}$

For most applications, we can conclude the decoding process assuming a suboptimal solution. We can also apply an additional step to get a maximum-likelihood decoding as follows.

Let $w_{\xi} \in \R^L$ be the closest point to $z_{\xi}$ in $\PP_{c_{\xi}}$ and \linebreak $y_{\xi} = \psi_{c_{\xi}}(w_{\xi})$ be its image in $S^{2L-1}$. If $d_{\xi} = ||y_{\xi} - x|| < \dfrac{d}{2}$, the maximum-likelihood decoding is over and $y_i$ is the solution for (\ref{decod}). 

If $d_{\xi} > \frac{d}{2}$, there might exist another $w$ in some other torus $T_{c_i}$ such that 
$$
||w-x||<||y_i-x||
$$

Let us define precisely what tori must be checked.

Let $ \NN = (\xi_i, \xi_2, \cdots, \xi_j)$, the set of tori for which 
$$
\Delta_i = ||x-\bar{x}_i|| < d_{\xi},
$$
We will assume that $\Delta_{\xi_i} \leq \Delta_{\xi_{i+1}} \ \ \forall i=1,2,...,j$.

Thus, we need to decode iteratively $x$ in the torus defined in $\NN$, getting a set of candidates $Y = \{y_{\xi}, y_{\xi_1}, y_{\xi_2}, \cdots, y_{\xi_j} \}$, $Y \subset TLSC(2k,d)$.

The output of decoding will be the point $y^* \in Y$ which satisfies 
$$
||y^* - x|| \leq ||y - x|| \ \ \forall y \in Y.
$$
In order to accelerate this process, each value of $d_{\xi}$, obtained iteratively, can be used to reduce the set $\NN$. Figure \ref{decodificacao} illustrates the decoding process in a $TLSC(2L,d)$. Each circle represents a torus $T_{c}$ in the code. In this example, just the tori $T_{c_{\xi}}$ and $T_{c_w}$ must be checked.

\begin{figure}[h!]
	\centering
		\includegraphics[width=8cm]{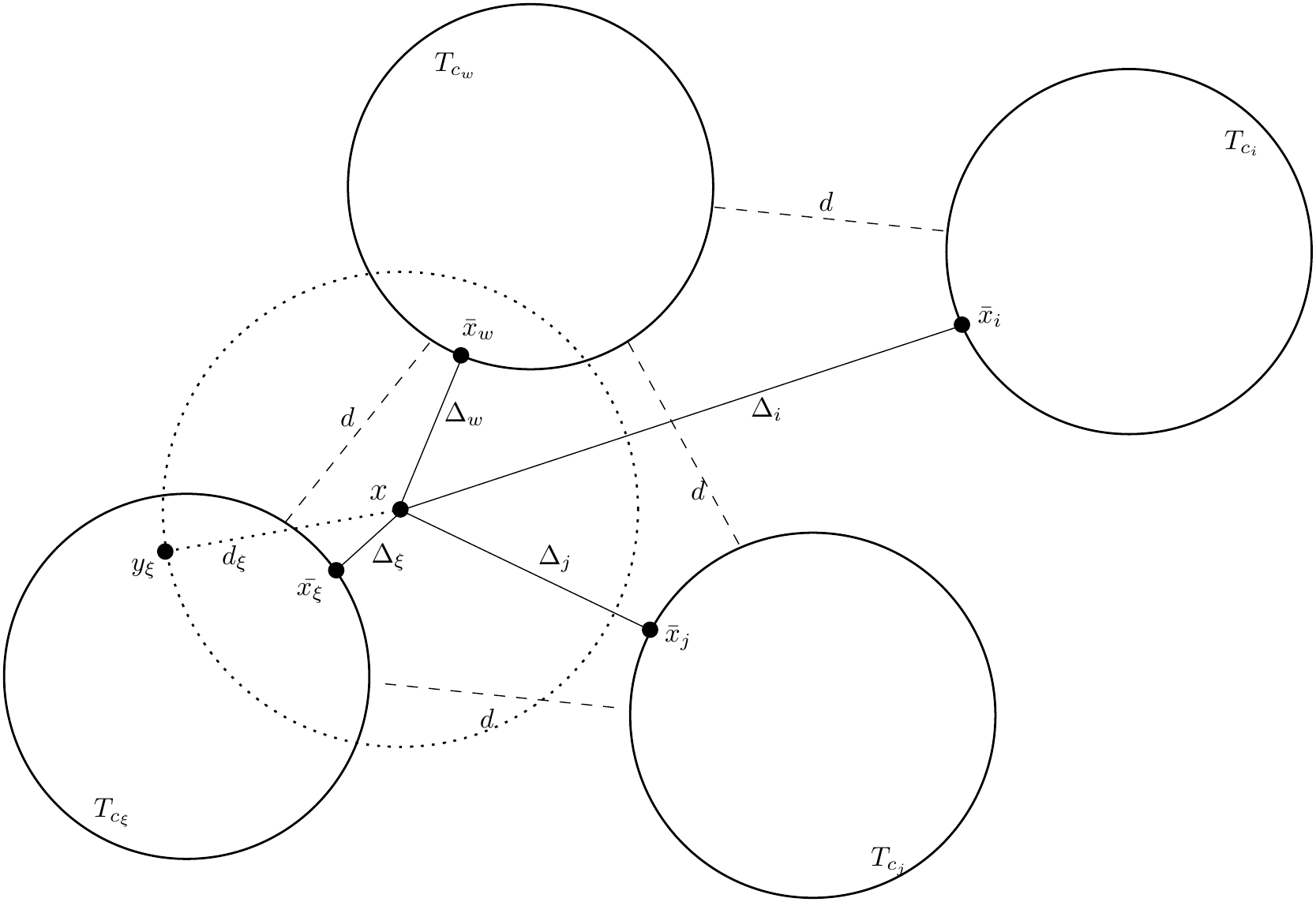}
		\caption{Decoding process in a $TLSC(2L,d)$.}
	\label{decodificacao}
\end{figure}

The computational complexity in a $2L$-dimensional spherical codes constructed in layers of flat tori is dominated by the complexity of decoding in a $L$-dimensional hyperbox.

\section{Conclusion}
We propose a new construction of spherical codes based on the foliation of the unit sphere in even dimensions by flat tori. Given a minimum distance $d$, the first step in this construction is to select torus layers which have minimum distance $d$. A codebook is then constructed in each layer by choosing a set of points in a hyperbox in half the code dimension. These points can be selected as cosets of a dense lattice in $R^L$ inducing a structured spherical code in $R^{48}$ which can be easily labeled and is generated by a commutative group of rotation matrix in each layer. The performance of these torus layer spherical codes is good when compared to the well-known wrapped spherical codes \cite{Hamkins1}, laminated spherical codes \cite{Hamkins2} and apple-peeling codes \cite{Gamal} for not asymptotically small distances. Concerning the coding and decoding process the main advantage comes from their homogeneous structure and the underlying lattice codebook in the half the code dimension.


\section*{Acknowledgment}
%C. Torezzan and S.I.R.Costa were partially supported by FAPESP (05/58102-7, 02/07473-7) and CNPq (304573/2002). 
The authors wish to thank N. J. A. Sloane for fruitful conversations and to AT\&T Shannon Laboratory where part of this work was done.

\ifCLASSOPTIONcaptionsoff
  \newpage
\fi


\begin{thebibliography}{1}
\bibitem{BergerGostiaux} M. Berger and B. Gostiaux, {\it Differential Geometry: Manifolds, Curves and Surfaces.} Berlin: Springer-Verlag, 1988.
\bibitem{BiglieriElia} E. Biglieri and M. Elia, ``Cyclic group codes for the Gaussian channel,'' IEEE Transactions on  Information Theory, vol.22, pp. 624-629, Sept. 1976.

\bibitem{Cohen} H. Cohen. A Course in Computational Algebraic Number Theory. Springer, 1996.

\bibitem{ConwaySloane} J. H. Conway and N. J. A. Sloane, {\it Sphere Packings, Lattices and Groups}. New York: {\it Springer}, 1998.



\bibitem{CMAP:2004} S. I. R. Costa, M. Muniz, E. Agustini and R. Palazzo, ``Graphs, tesselations and perfect codes on flat tori,'' IEEE Transactions on Information Theory, vol. 50, pp. 2363-2377, Oct. 2004.

\bibitem{Gamal} A. A. El Gamal, L. A. Hemachandra, I. Shperling and V.K. Wei, ``Using Simulated Annealing to Design Good Codes", IEEE Transactions on Information Theory, vol. IT-33, no. 1, pp. 116-123, Jan, 1987.
\bibitem{GershoGray} A. Gersho and R. M. Gray, {\it Vector Quantization and Signal Compression}. Boston, MA: Kluwer Academic, 1993.


\bibitem{Hamkins1} J. Hamkins and K. Zeger, ``Asymptotically dense spherical codes - Part I: Wrapped spherical codes,'' IEEE Transactions on Information Theory, vol. 43, pp. 1774-1784, Nov. 1997.

\bibitem{Hamkins2} J. Hamkins and K. Zeger, ``Asymptotically dense spherical codes - Part II: Laminated spherical codes,'' IEEE Transactions on Information Theory, vol. 43,  pp. 1786-1798, Nov. 1997.

\bibitem{kott} D. A. Kottwitz. ``The Densest Packing of Equal Circles on a Sphere''. Acta Cryst. A47 (1991), 158165.


\bibitem{SIQ} R. M.~Siqueira and S.~I.~R.~Costa, ``Flat Tori, Lattices and Bounds for Commutative Group Codes,''
Designs, Codes and Cryptography, vol. 49, pp. 307-312, Dec. 2008.
\bibitem{Slepian} D. Slepian, ``Group Codes for the Gaussian Channel,''  Bell Syst. Tech. J., vol. 47, pp. 575-602, 1968.
\bibitem{Sloane:Spherical} N. J. A. Sloane, ``Spherical Codes: Nice Arrangements of Points on a Sphere in Various Dimensions,'' [Online]. {\tt http://www.research.att.com/$\tilde{\ }$njas/packings/index.html}.
%\bibitem{Stillwell} J. Stillwell, {\it Geometry of Surfaces}. New York: Springer-Verlag, 1992.

\bibitem{zino} Ericson Th., Zinoviev V., Codes on Euclidean Spheres, North-Holland, Elsevier, 2001.

\bibitem{Torezzan} C. Torezzan, J.E. Strapasson, R. M. Siqueira, S.I.R. Costa, ``Optimum Commutative Group Codes''. Forthcoming paper.

\bibitem{Cos2003} V. Vaishampayan and S. I. R. Costa, ``Curves on a Sphere, Shift-Map Dynamics,
and Error Control for Continuos Alphabet Sources", IEEE
Transactions on Information Theory, vol. 49,  July 2003;

\bibitem{unger1} Ungerboeck, G. Channel coding with multilevel/phase signals. IEEE Transactions
on Information Theory 28, 1 (Jan 1982), 5567.

\end{thebibliography}
\end{document}